\documentclass[conference,10pt]{IEEEtran}
\usepackage{hyperref,latexsym,breakurl,amsmath,amssymb,color}
\usepackage{theorem}\theorembodyfont{\rmfamily}
\usepackage[shortlabels]{enumitem}
\setlist{nolistsep}  
\definecolor{darkred}{rgb}{.3, 0, 0}
\definecolor{darkgreen}{rgb}{0, .3, 0}
\newfont{\bbb}{bbm10 scaled 1000}                     
\newcommand{\IN}{\mbox{\bbb N}}                       
\newcommand{\IT}{\mbox{\bbb T}}                       
\newcommand{\T}{{\rm T}}                              
\newcommand{\R}{\mathrel{\cal R}}                     
\newcommand{\BRR}{\R}                                 
\newcommand{\BR}{\mathrel{\color{blue}\cal R}}        
\newtheorem{definition}{Definition}
\newtheorem{theorem}{Theorem}
\newtheorem{lemma}{Lemma}
\newtheorem{proposition}{Proposition}
\newtheorem{remark}{Remark}
\newtheorem{example}{Example}

\newcommand{\df}[1]{Def.~\ref{df:#1}}
\newcommand{\thm}[1]{Thm.~\ref{thm:#1}}
\newcommand{\pr}[1]{Prop.~\ref{pr:#1}}

\newenvironment{proof}{\begin{trivlist} \item[\hspace{\labelsep}\bf Proof:]}{\end{trivlist}}
\newcommand{\qed}{\hfill$\Box$}
\makeatletter
\def\comesfrom{\@transition\leftarrowfill}
\def\goesto{\@transition\rightarrowfill}
\def\ngoesto{\@transition\nrightarrowfill}
\def\Goesto{\@transition\Rightarrowfill}
\def\nGoesto{\@transition\nRightarrowfill}
\def\xmapsto{\@transition\mapstofill}
\def\nxmapsto{\@transition\nmapstofill}
\def\@transition#1{\@@transition{#1}}
\newbox\@transbox
\newbox\@arrowbox
\newbox\@downbox
\def\@@transition#1#2%
   {\setbox\@transbox\hbox
      {\vrule height 1.5ex depth .8ex width 0ex\hskip0.25em$\scriptstyle#2$\hskip0.25em}
   \ifdim\wd\@transbox<1.5em
      \setbox\@transbox\hbox to 1.5em{\hfil\box\@transbox\hfil}\fi
   \setbox\@arrowbox\hbox to \wd\@transbox{#1}
   \ht\@arrowbox\z@\dp\@arrowbox\z@
   \setbox\@transbox\hbox{$\mathop{\box\@arrowbox}\limits^{\box\@transbox}$}
   \dp\@transbox\z@\ht\@transbox 10pt
   \mathrel{\box\@transbox}}
\def\nrightarrowfill{$\m@th\mathord-\mkern-6mu%
  \cleaders\hbox{$\mkern-2mu\mathord-\mkern-2mu$}\hfill
  \mkern-6mu\mathord\not\mkern-2mu\mathord\rightarrow$}
\def\Rightarrowfill{$\m@th\mathord=\mkern-6mu%
  \cleaders\hbox{$\mkern-2mu\mathord=\mkern-2mu$}\hfill
  \mkern-6mu\mathord\Rightarrow$}
\def\nRightarrowfill{$\m@th\mathord=\mkern-6mu%
  \cleaders\hbox{$\mkern-2mu\mathord=\mkern-2mu$}\hfill
  \mkern-6mu\mathord\not\mathord\Rightarrow$}
\def\mapstofill{$\m@th\mathord\mapstochar\mathord-\mkern-6mu%
  \cleaders\hbox{$\mkern-2mu\mathord-\mkern-2mu$}\hfill
  \mkern-6mu\mathord\rightarrow$}
\def\nmapstofill{$\m@th\mathord\mapstochar\mathord-\mkern-6mu%
  \cleaders\hbox{$\mkern-2mu\mathord-\mkern-2mu$}\hfill
  \mkern-6mu\mathord\not\mkern-2mu\mathord\rightarrow$}
\makeatother 
\newcommand{\ar}[1]{\mathrel{\goesto{#1}}}            
\newcommand{\nar}[1]{{\ngoesto{#1\;}}}                
\newcommand{\goto}[1]{\stackrel{#1}{\longrightarrow}} 
\newcommand{\hoto}[1]{\mathbin{\stackrel{#1}{\raisebox{0pt} 
        [3pt][0pt]{$\scriptstyle--\rightarrow$}}}}          
\newcommand{\gonotto}[1]{\mbox{$\,\,\,\not\!\!\!\stackrel{#1~}{\longrightarrow}$}} 
\newcommand{\honotto}[1]{\mbox{$\,\,\,\,\not\!\!\!\!\stackrel{#1~}{\raisebox{0pt}
        [3pt][0pt]{$\scriptstyle--\rightarrow$}}$}}   
\newcommand{\plat}[1]{\raisebox{0pt}[0pt][0pt]{#1}}   
\newcommand{\rec}[1]{\plat{$			      
	\stackrel{\mbox{\tiny $/$}}
	{\raisebox{-.3ex}[.3ex]{\tiny $\backslash$}}
	\!\!#1\!\!
	\stackrel{\mbox{\tiny $\backslash$}}
	{\raisebox{-.3ex}[.3ex]{\tiny $/$}} $}}
\newcommand{\eqa}{\mathrel{\plat{$\stackrel{\alpha}=$}}} 
\newcommand{\Var}{{\it Var}}                          
\newcommand{\var}{{\it var}}                          
\newcommand{\PP}{P \vdash}
\newcommand{\weg}[1]{}                                
\title{Lean and Full Congruence Formats for Recursion}
\author{
\IEEEauthorblockA{\Large Rob van Glabbeek}
\IEEEauthorblockN{\small Data61, CSIRO, Sydney, Australia\\
   School of Computer Science \& Engineering,
   University of New South Wales, Sydney, Australia
}}

\begin{document}
\maketitle
\setcounter{footnote}{0}

\begin{abstract}
In this paper I distinguish two (pre)congruence requirements for semantic
equivalences and preorders on processes given as closed terms in a system
description language with a recursion construct.  A \emph{lean congruence}
preserves equivalence when replacing closed subexpressions of a process by
equivalent alternatives. A \emph{full congruence} moreover allows replacement
within a recursive specification of subexpressions that may contain recursion
variables bound outside of these subexpressions.

I establish that bisimilarity is a lean (pre)congruence for recursion for all
languages with a structural operational semantics in the ntyft/ntyxt format.
Additionally, it is a full congruence for the tyft/tyxt format.
\end{abstract}

\section{Introduction}

\emph{Structural Operational Semantics} \cite{Mi90ccs,Pl04} is one of the main methods for
defining the meaning of system description languages like CCS \cite{Mi90ccs}.
A system or \emph{process} is represented by a closed term built from a collection of
operators, process variables and usually a recursion construct, and the behaviour of a
process is given by its collection of (outgoing) transitions, each specifying the action
the process performs by taking this transition, and the process that results after doing so.
The transitions between states are obtained from a set of proof rules called \emph{transition rules}.

For purposes of representation and verification, several behavioural equivalence relations
have been defined on processes, of which the most well-known is \emph{(strong) bisimilarity}
\cite{Mi90ccs}. To allow compositional system verification, such equivalences
need to be \emph{congruences} for the operators under consideration, meaning that the
equivalence class of an $n$-ary operator $f$ applied to arguments $p_1,\ldots,p_n$ is
completely determined by the equivalence classes of these arguments.

Equally important is that the chosen equivalence relation $\sim$ is a congruence for recursion.
Recursion allows the specification of a process as a canonical solution of an equation $X=E(X)$.%
\footnote{The particular solution supplied by structural operational semantics is the one whose
  transitions are determined by the transition rules.}
Here $E(X)$ is an expression that may contain the variable $X$.
If $W$ is the collection of other variables occurring in $E(X)$, not \emph{bound}
by the recursive specification, then the canonical solution of $X=E(X)$
is a $W$-ary function that returns a process for each valuation of these
variables as processes. I call $\sim$ a \emph{lean congruence} for recursion if
each such operator satisfies the above-mentioned congruence requirement.

Take for example $E(X)$ to be $a.X + Y$ in the language CCS of {\sc Milner} \cite{Mi90ccs}.
Then $W=\{Y\}$. Let $\sim$ be bisimilarity, so that $b.0 \sim b.0 + b.0$ \cite{Mi90ccs}.
Now the lean congruence requirement for $\sim$ insists that the selected solutions of the recursive equations
$X = a.X + b.0$ and $X = a.X + (b.0 + b.0)$, obtained from $X=a.X+Y$ by substituting each of these
bisimilar processes for $Y$, are again bisimilar.

The lean congruence requirement plays a key r\^ole in the study of expressiveness of system description languages~\cite{vG12}.
There, \emph{correct translations} of one language into another up to a semantic equivalence $\sim$ are defined;
and expressiveness hierarchies---one for each choice of $\sim$---are defined in terms of those translations. 
However, a correct translation can exist only when $\sim$ is a lean congruence for the source
language, as well as for the source's image within the target language.

If $F(X)$ is an expression like $E(X)$, for simplicity assuming that neither
contains variables other than $X$, and $E(p) \sim F(p)$ regardless which process $p$ is
substituted for the variable $X$, then the \emph{full} congruence property
demands that the selected solutions of the equations $X=E(X)$ and $X=F(X)$ are again equivalent.
As a CCS example, suppose that a process is given as the solution of
the equation $X=a.X + a.X$. Using the idempotence of $+$ under bisimilarity, one
can now proceed to think of the same process, up to bisimilarity, as the solution of $X=a.X$. This type of
reasoning is a central component in system verification by equivalence checking \cite{Ba90,Fok00,BBR10,GM14},
as applied in successful verification toolsets such as
\href{http://cadp.inria.fr/}{CADP} \cite{GLMS11} and
\href{http://www.win.tue.nl/mcrl2/}{mCRL2} \cite{GM14}.
Yet it is valid only if bisimilarity is a full congruence for recursion.

In order to streamline the process of proving that a certain equivalence is a congruence
for certain operators, and to guide sensible language definitions, syntactic criteria
(\emph{congruence formats}) for the transition rules in structural operational semantics
have been developed, ensuring that the equivalence is a congruence for any operator
specified by rules that meet these criteria.  The first of these was proposed by
{\sc Robert de Simone} in \cite{dS84,dS85} and is now called the \emph{De Simone format}.
A  generalisation featuring transition rules with negative premises is the
\emph{GSOS format} of {\sc Bloom, Istrail \& Meyer} \cite{BIM95}, and a generalisation with
\emph{lookahead} is the \emph{tyft/tyxt} format of {\sc Groote \& Vaandrager} \cite{GrV92}.
The \emph{ntyft/ntyxt} format of {\sc Groote} \cite{Gr93} allows both
negative premises and lookahead and generalises the GSOS as well as the tyft/tyxt format.
All this work provides congruence formats for (strong) bisimilarity.
Congruence formats for other \emph{strong} semantic equivalences---treating the internal
action $\tau$ like any other action---appear in \cite{BFG04,FG17}.\footnote{These congruence formats also
apply to behavioural \emph{preorders}, and then ensure that such a preorder is a \emph{precongruence}.}
Formats for \emph{weak} semantics---abstracting from internal activity---can be found, e.g., in 
\cite{Ul92a,Bl95,Fok00b,Ul00,UP02,vG11,FGW12,FG16a}.

Extensions to probabilistic systems appear for instance in
\cite{Bartels02,LT09,KS13,GS13,MM14,BM15,DGL16}.
Rule formats ensuring properties of operators other than being a (pre)congruence appear in
\cite{MRG05} (commutativity), \cite{CMR08} (associativity), \cite{ACIMR11} (zero and unit elements),
\cite{ACIMR12} (distributivity) and \cite{ABIMR12} (idempotence). 
Overviews on work on congruence formats and other rule formats, with many more references, can be
found in \cite{AFV00,GMR06}.

Yet, to the best of my knowledge, no one has proposed a congruence format for recursion.
This hiatus is addressed here.
I establish that bisimilarity is a lean congruence for recursion for all languages
with a structural operational semantics in the ntyft/ntyxt format.\footnote{Some of those languages
  have a 3-valued transition system semantics, where bisimilarity becomes an asymmetric preorder.
  Here I establish that it is a precongruence.}
I did not succeed in showing that it is even a full congruence for all ntyft/ntyxt languages;
nor did I find a counterexample. Even for GSOS languages this remains an open question.
However, I show that bisimilarity is a full congruence for recursion for all tyft/tyxt languages.

My proof strategy follows the traditional method of \cite{BIM95,GrV92,BolG96}.
However, for this to work smoothly, I present a new formulation---better fitted to my
application---of the well-founded semantics of transition system specifications with negative
premises, and show its consistency with previous formulations.

I could not establish the full congruence result directly, without using the lean
congruence result as an intermediate step, even when restricting the latter to the tyft/tyxt format.
Thus, I see no way around a sequence of two proofs with a large overlap.

The method of \emph{modal decomposition} \cite{FGW06} yields alternative congruence proofs
for operators specified in the tyft/tyxt and GSOS formats \cite{FGW06}.
Extending this method to deal with recursion might be a way to extend my full
congruence result to transition rules with negative premises.

Providing (lean and full) congruence formats for recursion for equivalences and preorders other than
bisimilarity, as well as for weak versions of bisimilarity \cite{Mi90ccs,vGW89}---supporting
abstraction from internal actions---remains an important open problem.

\section{Transition system specifications and their meaning}

In this paper $\Var$ and $A$ are two sets of {\em variables} and {\em
actions}. Many concepts that will appear are parameterised by the
choice of $\Var$ and $A$, but as in this paper this choice is fixed, a
corresponding index is suppressed.

\begin{definition}[\emph{Signatures}]\label{df:signature}\rm
A {\em function declaration} is a pair $(f,n)$ of a {\em function
symbol} $f \not\in \Var$ and an {\em arity} $n \in \IN$.\footnote{This work generalises seamlessly to
  operators with infinitely many arguments. Such operators occur, for instance, in \cite[Appendix A.2]{BrGH16b}.
  Hence one may take $n$ to be any ordinal.
  An operator, like the \emph{summation} or \emph{choice} of CCS \cite{Mi90ccs}, that actually takes any \emph{set}
  of arguments, needs to be simulated by a family of operators with a \emph{sequence} of arguments (but
  yielding the same value upon reshuffling of the arguments), one for each cardinality of this set.} 
A function declaration $(c,0)$ is also called a {\em constant declaration}.
A {\em signature} is a set of function declarations. The set
$\IT(\Sigma)$ of {\em terms with recursion} over a signature $\Sigma$ is defined
inductively by:
\begin{itemize}
\item $\Var \subseteq \IT(\Sigma)$,
\item if $(f,n) \mathbin\in \Sigma$ and $t_1,
  ...,t_n \mathbin\in \IT(\Sigma)$ then
$f(t_1,...,t_n) \mathbin\in \IT(\Sigma)$,
\item If $V_S \subseteq \Var$, $~S:V_S \rightarrow \IT(\Sigma)$ and $X\in V_S$,
then $\rec{X|S}\in \IT(\Sigma)$.
\end{itemize}
A term $c()$ is abbreviated as $c$. A function $S$ as appears in
the last clause is called a {\em recursive specification}.  A
recursive specification $S$ is often displayed as $\{X\mathbin=S_X \mid X \mathbin\in V_S\}$.
An occurrence of a variable $y$ in a term $t$ is {\em free} if it does not
occur in a subterm of $t$ of the form $\rec{X|S}$ with $y \in V_S$.
Let $\var(t)$ denote the set of variables occurring free in a
term $t\in\IT(\Sigma)$, and let $\IT(\Sigma,W)$ be
the set of terms $t$ over $\Sigma$ with $\var(t)\subseteq W$.
$\T(\Sigma):=\IT(\Sigma,\emptyset)$ is set of \emph{closed} terms over $\Sigma$.
\end{definition}

\begin{example}\label{ex:recursion}
Let $\Sigma$ contain three unary functions $a.\_$, $b.\_$ and $d.\_$, and one infix-written binary
function $\|$. Let $X,Y,z\mathbin\in\Var$. Then $S=\{~X=(a.X)\|(b.Y),\quad  Y=(d.Y)\|(X\|z)~\}$ is a recursive
specification, so $\rec{X|S}\mathbin\in\IT(\Sigma)$. Since $V_S\mathbin=\{X,Y\}$, the only variable that occurs free
in this term is $z$.
\end{example}
As illustrated here, I often choose upper case letters for bound variables (the ones occurring in a set
$V_S$) and lower case ones for variables occurring free; this is a convention only.

A recursive specification $S$ is meant to denote a $V_S$-tuple (in the example above a pair) of
processes that---when filled in for the variables in $V_S$---forms a solution to the equations
in $S$.\footnote{When $S$ contains free variables from a set $W$, this solution is parameterised by the choice of a
valuation of these variables as processes, thereby becoming a $W$-ary function.}\linebreak[4]
The term $\rec{X|S}$ denotes the $X$-component of such a tuple.

\begin{definition}[\emph{Substitution}]\label{df:substitutions}\rm
A {\em $\Sigma$-substitution} $\sigma$ is a partial function from $\Var$ to
$\IT(\Sigma)$; it is \emph{closed} if it is a total function from $\Var$ to $\T(\Sigma)$.
If $\sigma$ is a substitution and $S$ any syntactic
object, then $S[\sigma]$ denotes the object obtained from $S$ by
replacing, for $x$ in the domain of $\sigma$, every free occurrence of $x$
in $S$ by $\sigma(x)$, while renaming bound variables if necessary to prevent
name-clashes. In that case $S[\sigma]$ is called a {\em substitution instance} of $S$.
A substitution instance $t[\sigma]$ where $\sigma$
is given by $\sigma(x_i)=u_i$ for $i\in I$ is denoted as $t[u_i/x_i]_{i\in I}$,
and for $S$ a recursive specification $\rec{t|S}$ abbreviates $t[\rec{Y|S}/Y]_{Y\in V_S}$.
\end{definition}
\begin{example}
Extend $\Sigma$ from Ex.~\ref{ex:recursion} with a constant $c$. Then
$\rec{X|S}[b.c/z] = \rec{X|\{X{=}(a.X)\|(b.Y),\;  Y{=}(d.Y)\|(X\|b.c)\}}$,
$\rec{X|S}[X/z] = \rec{Z|\{Z{=}(a.Z)\|(b.Y),\quad  Y{=}(d.Y)\|(Z\|X)\}}$ and
$\rec{X|S}[b.c/Y] = \rec{X|S}$.
\end{example}

\noindent
Structural operational semantics \cite{Pl04} defines the meaning of system description languages
whose syntax is given by a signature $\Sigma$. It generates a transition system in which the
states, or \emph{processes}, are the closed terms over $\Sigma$---representing the remaining system
behaviour from that state---and transitions between processes are supplied with labels. The
transitions between processes are obtained from a transition system specification, which consists of
a set of transition rules.

\begin{definition}[\emph{Transition system specifications}]\label{df:TSS}\rm
Let $\Sigma$ be a signature. A {\em positive $\Sigma$-literal} is an
expression \plat{$t \ar{a} t'$} and a {\em negative $\Sigma$-literal} an
expression \plat{$t \nar{a}$} with $t,t'\in\IT(\Sigma)$ and $a \in A$.
For $t,t' \in \IT(\Sigma)$ the literals $t \ar{a} t'$ and $t \nar{a}$
are said to {\em deny} each other.
A {\em transition rule} over $\Sigma$ is an expression of the form
$\frac{H}{\alpha}$ with $H$ a set of $\Sigma$-literals (the {\em
premises} or {\em antecedents} of the rule) and $\alpha$ a
positive $\Sigma$-literal (the {\em conclusion}).
The terms at the left- and right-hand side of $\alpha$ are
the \emph{source} and \emph{target} of the rule.
A rule \plat{$\frac{H}{\alpha}$} with $H=\emptyset$ is also written $\alpha$.
A literal or transition rule is {\em closed} if it contains no free variables.
A {\em transition system specification (TSS)} is a pair $(\Sigma,R)$
with $\Sigma$ a signature and $R$ a set of transition rules over $\Sigma$; it is
{\em positive} if all antecedents of its rules are positive.
\end{definition}
The concept of a (positive) TSS presented above was introduced in
{\sc Groote \& Vaandrager} \cite{GrV92}; the negative premises \plat{$t\nar{a}$}
were added in {\sc Groote} \cite{Gr93}. The notion
generalises the {\em GSOS rule systems} of \cite{BIM95} and constitutes the
first formalisation of {\sc Plotkin}'s {\em Structural Operational
Semantics (SOS)} \cite{Pl04} that is sufficiently general to cover
many of its applications.

The following definition (from \cite{vG93d}) tells when a transition is provable from a
TSS\@. It generalises the standard definition (see e.g.\ \cite{GrV92})
by (also) allowing the derivation of transition rules. The
derivation of a transition \plat{$t\ar{a}t'$} corresponds to the derivation
of the transition rule $\frac{H}{t\goto{a}t'}$ with $H\mathbin=\emptyset$.
The case $H \mathbin{\neq} \emptyset$ corresponds to the derivation of
\plat{$t\ar{a}t'$} under the assumptions $H$.

\begin{definition}[\emph{Proof}]\label{df:proof}\rm
Let $P=(\Sigma,R)$ be a TSS. A {\em proof} of a transition
rule $\frac{H}{\alpha}$ from $P$ is a well-founded, upwardly
branching tree of which the nodes are labelled by $\Sigma$-literals,
such that:
\begin{itemize}
\item the root is labelled by $\alpha$, and
\item if $\beta$ is the label of a node $q$ and $K$ is the set of
labels of the nodes directly above $q$, then
\begin{itemize}
\item either $K=\emptyset$ and $\beta \in H$,
\item or $\frac{K}{\beta}$ is a substitution instance of a rule from $R$.
\end{itemize}
\end{itemize}
If a proof of $\frac{H}{\alpha}$ from $P$ exists, then\vspace{2pt}
$\frac{H}{\alpha}$ is {\em provable} from $P$, notation $P \vdash
\frac{H}{\alpha}$.
\end{definition}

\noindent
A TSS is meant to specify an LTS in which the transitions are closed positive literals.
A positive TSS specifies a transition relation in a straightforward
way as the set of all provable transitions.\footnote{Readers interested only in the restriction of
  my results to TSSs without negative premises---giving rise to 2-valued transition relations---can
  safely skip the remainder of this section, and identify $p\hoto{a}p'$ with $p\goto{a}p'$.
  In the proofs of \pr{alpha-conversion bisimilar} and \thm{congruence} also $p\hoto{a}_\lambda p'$
  and $p\goto{a}_\lambda p'$ equal $p\goto{a}p'$, for any $\lambda$; so the induction on $\lambda$
  can be skipped, as well as the auxiliary Claims 3 and 1, and the proof proceeds directly by induction on $\pi$.}
But as pointed out in {\sc Groote} \cite{Gr93}, it is not so easy to associate a
transition relation to a TSS with negative premises.
In \cite{vG04} several solutions to this problem were reviewed and evaluated.
Arguably, the best method to assign a meaning to all TSSs is the \emph{well-founded semantics}
of {\sc Van Gelder, Ross \& Schlipf} \cite{GRS91}, which in general yields a
{\em 3-valued transition relation}
$T: {\sf T}(\Sigma) \times A \times {\sf T}(\Sigma)\rightarrow
\{\mbox{\sl present}, \mbox{\sl undetermined}, \mbox{\sl absent}\}$.
I present such a relation as a pair $\rec{CT,PT}$ of 2-valued
transition relations---the sets of \emph{certain} and \emph{possible transitions}---with $CT \subseteq PT$.
When insisting on 2-valued transition relations, the best method is the same,
declaring meaningful only those TSSs whose well-founded semantics is 2-valued, meaning that $CT=PT$.

Below I give a new presentation of the well-founded semantics, strongly inspired by
previous accounts in \cite{Prz90,BolG96,vG04}. As \df{proof} does not allow the
derivation of negative literals, to arrive at an approximation $AT^+$ of the set of transitions
that are in the transition relation intended by a TSS $P$, one could start from an approximation
$AT^-$ of the closed negative literals that ought to be generated, and define $AT^+$ as the
set of closed positive literals provable from $P$ under the hypotheses $AT^-$.
Intuitively,
\begin{enumerate}
  \item if $AT^-$ is an under- (resp.\ over-)approximation of the closed negative
    literals that ``really'' hold, then $AT^+$ will be an under- (resp.\ over-)approximation
    of the intended (2-valued) transition relation, and
  \item if $AT^+$ is an under- (resp.\ over-)approximation of the intended transition
    relation, then the set of all closed negative literals that do not deny any literal in
    $AT^+$ is an over- (resp.\ under-)approximation of the closed negative literals
    that agree with the intended transition relation.
\end{enumerate}

{\hfuzz 1.6pt

\begin{definition}[\emph{Over- and underappr.\ 
      of transition relations}]
\label{df:well-founded}
Let $P$ be a TSS\@.
For ordinals $\lambda$ the sets $CT^+_\lambda$ and $PT^+_\lambda$
of closed positive literals, and $CT^-_\lambda$, $PT^-_\lambda$
of closed negative literals are defined inductively~by:

\noindent
\hfill
$PT^-_\lambda$ \begin{tabular}{l}~\\is the set of literals\\ that do not
  deny any\\ $\beta\in CT^+_\kappa$ with $\kappa<\lambda$
  \end{tabular}
\hfill
$\beta \in PT^+_\lambda$ iff $P\vdash \frac{PT^-_\lambda}{\beta}$\phantom.
\hfill\mbox{}

\noindent
\hfill
$CT^-_\lambda$ \begin{tabular}{l}~\\is the set of literals\\ that do not
  deny any\\ $\beta\in PT^+_\lambda$
  \end{tabular}
\hfill
$\beta \in CT^+_\lambda$ iff $P\vdash \frac{CT^-_\lambda}{\beta}$.
\hfill\mbox{}
\end{definition}}
Intuitively, $CT^+_\lambda$ is an underapproximation of the set of transitions that should
be in the transition relation specified by $P$, and $PT^+_\lambda$ an overapproximation.
Likewise, $CT^-_\lambda$ is an underapproximation of the set of closed negative literals
that should hold, and $PT^-_\lambda$ an overapproximation.
The approximations get better with increasing $\lambda$.
To understand this inductively, note that $PT^-_0$ is the set of \emph{all}
closed negative literals, and thus surely an overapproximation.
The induction step is given by considerations 1 and 2 above.

\begin{lemma}
$CT^-_\kappa \mathbin\subseteq CT^-_\lambda \mathbin\subseteq PT^-_\lambda \mathbin\subseteq PT^-_\kappa$
and $CT^+_\kappa \mathbin\subseteq CT^+_\lambda \mathbin\subseteq PT^+_\lambda \mathbin\subseteq PT^+_\kappa$
for~$\kappa\mathbin<\lambda$.
\end{lemma}
\begin{proof}
  Let $\kappa<\lambda$. The definition of $PT_\lambda^-$ immediately yields $PT^-_\lambda \mathbin\subseteq PT^-_\kappa$.
  From this, applying \df{well-founded}, one obtains $PT^+_\lambda \mathbin\subseteq PT^+_\kappa$,
  $CT^-_\kappa \mathbin\subseteq CT^-_\lambda$ and $CT^+_\kappa \mathbin\subseteq CT^+_\lambda$, respectively.
  The remaining claims follow by induction on $\lambda$.

  As $PT^-_0$ is the universal relation, certainly $CT^-_0 \mathbin\subseteq PT^-_0$,
  so $CT^+_0 \mathbin\subseteq PT^+_0$.

  Let $\lambda$ be a limit ordinal. Then $PT^-_\lambda = \bigcap_{\mu<\lambda} PT^-_\mu$.
  For any $\kappa,\mu<\lambda$ one has $CT^-_\kappa \subseteq PT^-_\mu$ by induction.
  Namely $CT^-_\kappa \subseteq CT^-_\mu \subseteq PT^-_\mu$ if $\kappa \leq \mu <\lambda$,
  and $CT^-_\kappa \subseteq PT^-_\kappa \subseteq PT^-_\mu$ if $\mu \leq \kappa <\lambda$. 
  Hence \plat{$CT^-_\kappa \subseteq \bigcap_{\mu<\lambda} PT^-_\mu = PT^-_\lambda$} for any $\kappa<\lambda$,
  and hence $CT^+_\kappa \subseteq PT^+_\lambda$. With \df{well-founded} this implies 
  $CT^-_\lambda \mathbin\subseteq PT^-_\lambda$ and hence $CT^+_\lambda \mathbin\subseteq PT^+_\lambda$.

  Now let $\lambda=\mu{+}1$. By induction $CT^+_\mu \mathbin\subseteq PT^+_\mu$.
  With \df{well-founded} this implies $CT^-_\mu \mathbin\subseteq PT^-_\lambda$,
  and hence $CT^+_\mu \subseteq PT^+_\lambda$. With \df{well-founded} this implies 
  $CT^-_\lambda \mathbin\subseteq PT^-_\lambda$ and hence $CT^+_\lambda \mathbin\subseteq PT^+_\lambda$.
\qed \end{proof}
Since the closed literals over $\Sigma$ form a proper set, there must be an ordinal $\kappa$ such
that $PT^-_\lambda \mathbin= PT^-_\kappa$ for all $\lambda\mathbin>\kappa$, and hence also $PT^+_\lambda \mathbin= PT^+_\kappa$,
$CT^-_\lambda \mathbin= CT^-_\kappa$ and $CT^+_\lambda \mathbin= CT^+_\kappa$.
\begin{definition}\rm\label{df:closure ordinal}
Such an ordinal $\kappa$ is called \emph{closure ordinal}. Let $PT^-\mathbin{:=}PT^-_\kappa$,
$PT^+\mathbin{:=}PT^+_\kappa$, $CT^-\mathbin{:=}CT^-_\kappa$ and $CT^+\mathbin{:=}CT^+_\kappa$.%
\end{definition}
\begin{remark}\label{union}
$PT^-\mathbin=\bigcap_\lambda PT^-_\lambda$, taking the intersection over all ordinals.
Likewise, $PT^+\mathbin=\bigcap_\lambda PT^+_\lambda$,
$CT^-\mathbin= \bigcup_\lambda CT^-_\lambda$ and $CT^+=\bigcup_\lambda CT^+_\lambda$.
\end{remark}
\begin{remark}\label{determined}
$PT^-$ is the set of literals that do not deny any literal in $CT^+$, and likewise for
$CT^-$ and $PT^+$. Moreover, $CT^- \subseteq PT^-$ and $CT^+ \subseteq PT^+$.
\end{remark}
\begin{definition}[\emph{Well-founded semantics}]\rm\label{df:well-founded semantics}
The 3-valued transition relation $\rec{CT^+\!,PT^+}$ constitutes the well-founded semantics of $P$.
\end{definition}
Below I show that the above account of the well-founded
semantics is consistent with the one in \cite{vG04}, and thereby with the ones in \cite{BolG96,Prz90,GRS91}.

\begin{definition}[\emph{Well-supported proof} \cite{vG04}]\label{df:wsp}\rm
Let $P=(\Sigma,R)$ be a TSS\@. A \emph{well-supported proof} from $P$ of a closed literal ${\alpha}$ is a well-founded tree with the nodes labelled by closed literals, such that the root is labelled by $\alpha$, and if $\beta$ is the label of a node and $K$ is the set of labels of the children of this node, then:
\begin{itemize}
\item either $\beta$ is positive and $\frac{K}{\beta}$ is a substitution instance of a rule in $R$;
\item or $\beta$ is negative and for each set $N$ of closed negative literals with
  \plat{$P \vdash \frac{N}{\gamma}$} for
  $\gamma$ a closed positive literal denying $\beta$, a literal in $K$ denies one in $N$.
\end{itemize}
$P\vdash_{\it ws}\alpha$ denotes that a well-supported proof from $P$ of $\alpha$ exists.
\end{definition}

\begin{proposition}
Let $P$ be a TSS\@. Then $P \vdash_{\it ws}p\ar a q$ iff $(p\ar a q)\in CT^+$, and $P
\vdash_{\it ws}p\nar a$ iff $(p\nar a)\in CT^-$.
\end{proposition}

\begin{proof}
\framebox{$\Rightarrow$}: Let $\pi$ be a well-supported proof of a closed literal
$\alpha$. By consistently applying the same closed substitution to all literals occurring
in $\pi$, one can assume, without loss of generality, that all literals in $\pi$ are closed.
With structural induction on $\pi$ I show that $\alpha \in CT^+ \cup CT^-$.

Suppose $\alpha$ is positive and $\frac{K}{\alpha}$ is the closed substitution instance of
the rule of $P$ applied at the root of $\pi$. Then for each $\beta\in K$ the literal
$\beta$ is ws-provable from $P$ by means of a strict subproof of $\pi$. By induction 
$\beta \in CT^+ \cup CT^-$. As $CT^+$ is $CT^+_\kappa$ for some ordinal $\kappa$, it is
closed under deduction. Hence $\alpha \in CT^+$.

Suppose $\alpha$ is negative. Let $\beta$ be closed positive literal denying $\alpha$.
By \df{wsp}, each set $N$ of closed negative literals with
$P\vdash\beta$ contains a literal $\gamma_N$ denying a literal $\delta_N$ that is ws-provable
from $P$ by means of a strict subproof of $\pi$. By induction $\delta_N \mathbin\in CT^+$.
Hence $\gamma_N \mathbin{\notin} PT^-$. Consequently $\beta \mathbin{\notin} PT^+$. Hence $\alpha \mathbin\in CT^-$.

\framebox{$\Leftarrow$}: Suppose $\alpha \in CT_\lambda^+ \cup CT^-_\lambda$.
With induction on $\lambda$ I show that $P \vdash_{\it ws}\alpha$.
First suppose $\alpha \in CT_\lambda^-$. Let $N$ be a set of closed negative literals with
\plat{$P \vdash \frac{N}{\gamma}$} for $\gamma$ a closed positive literal denying $\alpha$.
Assume that $N \subseteq PT^-_\lambda$. Then $\gamma$ would be in $PT^+_\lambda$,
contradicting the definition of $CT^-_\lambda$. So $N$ contains a literal that is not in
$PT^-_\lambda$, i.e., denies a literal $\delta_N$ in $CT^+_\kappa$ for some $\kappa < \lambda$.
By induction, $P \vdash_{\it ws} \delta_N$. It follows that $P \vdash_{\it ws}\alpha$.

Now suppose $\alpha \in CT_\lambda^+$. Then \plat{$P \vdash \frac{CT^-_\lambda}{\alpha}$}.
By the case above $P \vdash_{\it ws}\beta$ for each $\beta\in CT^-_\lambda$.
Hence $P \vdash_{\it ws}\alpha$.
\qed
\end{proof}
The above result, together with Theorem 1 in \cite{vG04}, and the observation in
\cite{vG04} that literals \plat{$t \mathrel\nar{a} t'$} can be eliminated from consideration
(as done here), implies that the well-founded semantics given above agrees with the
one from~\cite{vG04}.

In \cite{vG04} it was shown that $\vdash_{\it ws}$ is consistent, in the sense that no TSS
admits well-supported proofs of two literals that deny each other. This also follows
directly from the material above. 
A TSS $P$ is called \emph{complete} \cite{vG04} if for each $p$ and $a$, either
\plat{$P\vdash_{\it ws}p\nar a$} or \plat{$P\vdash_{\it ws}p\ar a q$} for some $q$.
This implies that $CT^-$ is exactly the set of closed negative literals that do not deny any
literal in $CT^+$. Hence $CT^-=PT^-$ and thus $CT^+=PT^+$. So the 3-valued transition
system associated to a complete TSS is 2-valued.

Below I write $P \vdash p \goto{a}_\lambda q$ for $(p \mathbin{\goto{a}} q)\mathbin\in CT^+_\lambda$,
$P \vdash p \gonotto{a}_\lambda$ for $(p \gonotto{a})\mathbin\in CT^-_\lambda$,
$P \vdash p \hoto{a}_\lambda q$ for $(p \goto{a} q)\in PT^+_\lambda$ and
$P \vdash p \honotto{a}_\lambda$ for $(p \gonotto{a})\in PT^-_\lambda$.
Moreover, $p \goto{a} q$, resp.\ $p \hoto{a} q$, will abbreviate $p \goto{a}_\kappa q$,
resp.\ $p \hoto{a}_\kappa q$, where $\kappa$ is the closure ordinal of \df{closure ordinal}.

In my forthcoming lean congruence proof I will apply structural induction on ``the proof of a transition
$p \goto{a}_\lambda q$ or $p \hoto{a}_\lambda q$ from $P$''.
There I will mean the proofs of $\frac{CT^-_\lambda}{p \goto{a} q}$\vspace{-4pt} and
$\frac{PT^-_\lambda}{p \goto{a} q}$, respectively, as this is what constitutes the evidence for the
statement \plat{$P \vdash p \ar{a}_\lambda q$}, resp.\ \plat{$P \vdash p \hoto{a}_\lambda q$}.

\section{The bisimulation preorder}

The goal of this paper is to show that bisimilarity is a congruence for recursion for all languages
with a structural operational semantics in the ntyft/ntyxt format.
Traditionally \cite{Mi90ccs}, bisimilarity is defined on 2-valued transition systems only, whereas
the structural operational semantics of a language specified by a TSS
can be 3-valued. Rather than limit my results to languages specified
by complete TSSs, I use an extension of the notion of bisimilarity to 3-valued transition systems.
Such an extension, called \emph{modal refinement}, is provided in \cite{LT88}.
There, 3-valued transition systems are called \emph{modal transition systems}.

\begin{definition}[\emph{Bisimilarity}]\rm\label{df:refinement}
Let $P$ be a TSS\@.
A \emph{bisimulation} $\R$ is a binary relation on the states of $\T(\Sigma)$ such that, for
$p,q\in\T(\Sigma)$ and $a\in A$,
\begin{itemize}
\item if $p\R q$ and $\PP p \goto{a} p'$, then there is a $q'$ with $\PP q \goto{a} q'$ and $p' \R q'$,
\item if $p\R q$ and $\PP q \hoto{a} q'$, then there is a $p'$ with $\PP p \hoto{a} p'$ and $p' \R q'$.
\end{itemize}
A process $q\mathbin\in \T(\Sigma)$ is a \emph{modal refinement} of $p\mathbin\in \T(\Sigma)$, notation
$p \sqsubseteq_B q$, if there exists a bisimulation $\R$ with $p \R q$.
I call $\sqsubseteq_B$ the \emph{bisimulation preorder}, or \emph{bisimilarity}.
The kernel of $\sqsubseteq$, given by ${\equiv_B} := {\sqsubseteq_B} \cap {\sqsupseteq_B}$, is \emph{bisimulation equivalence}.
\end{definition}
Clearly, modal refinement is reflexive and transitive, and hence a preorder.
The underlying idea is that a process $p$ with a 3-valued transition relation $\rec{CT,PT}$ is a \emph{specification} of a
process with a 2-valued transition relation, in which the presence or absence of certain transitions is left open.
$CT$ contains the transitions that are \emph{required} by the specification, and $PT$ the ones that are \emph{allowed}.
If $p \sqsubseteq_B q$, then $q$ may be closer to the eventual implementation, in the sense that
some of the undetermined transitions have been resolved to present or absent.
The requirements of \df{refinement} now say that any transition that is required by $p$ should be
(matched by a transition) required by $q$, whereas any transition allowed by $q$, should certainly
be (matched by a transition) allowed by $p$.

In case $p$ and $q$ are 2-valued (i.e.~\emph{implementations}) the modal refinement relation is
just the traditional notion of bisimilarity \cite{Mi90ccs} (and thus symmetric).

While achieving a higher degree of generality of my lean congruence theorem by interpreting
incomplete TSSs as modal transition systems, I do not propose incomplete TSSs as a tool for
the specification of modal transition systems.

\section{Congruence properties}

In the presence of recursion, two sensible notions of precongruence come to mind.
Let $\sqsubseteq$ be a preorder on the set $\T(\Sigma)$ of closed terms over 
$\Sigma$.
For $\rho,\nu\!:\!\Var\rightarrow\T(\Sigma)$ closed substitutions write $\rho\mathbin\sqsubseteq \nu$ iff
$\rho(x) \mathbin\sqsubseteq \nu(x)$ for each $x \mathbin\in \Var$.

\begin{definition}[\emph{Lean precongruence}]\rm\label{df:lean congruence}
A preorder ${\sqsubseteq} \subseteq \T(\Sigma)\linebreak\times\T(\Sigma)$ is a \emph{lean precongruence} iff $t[\rho] \sqsubseteq t[\nu]$ for any term
$t\in\IT(\Sigma)$ and any closed substitutions $\rho$ and $\nu$ with $\rho \sqsubseteq \nu$.
\end{definition}

\begin{definition}[\emph{Full precongruence}]\rm\label{df:full congruence}
A preorder ${\sqsubseteq} \mathbin\subseteq \T(\Sigma)\times\T(\Sigma)$ is a \emph{full precongruence} iff
it satisfies\vspace{-1ex}%
\begin{equation}\label{comp-operators-closed}
\begin{array}{l}p_i\sqsubseteq q_i ~\mbox{for all}~i=1,...,n \\~~\Rightarrow~~
 f(p_1,...,p_n) \sqsubseteq f(q_1,...,q_n)\end{array}
\vspace{-2ex}
\end{equation}
\begin{equation}\label{comp-recursion-closed}
\begin{array}{r}
S_Y[\sigma]\sqsubseteq S_Y'[\sigma] ~\mbox{for all}~Y\mathbin\in W~\mbox{and}~ \sigma\!:W \rightarrow \T(\Sigma)
\\\Rightarrow~~ \rec{X|S} \sqsubseteq \rec{X|S'}\end{array}
\end{equation}
for all functions $(f,n)\in\Sigma$, closed terms $p_i,q_i\in \T(\Sigma)$, and
recursive specifications $S,S':W \rightarrow \IT(\Sigma,W)$ with $X \in W \subseteq \Var$.
\end{definition}
A lean (resp.\ full) precongruence that is symmetric (i.e.\ an equivalence relation) is called a
\emph{lean} (resp.\ \emph{full}) \emph{congruence}.
Clearly, each full (pre)congruence is also a lean (pre)congruence, and each lean (pre)congruence satisfies
(\ref{comp-operators-closed}) above. Both implications are strict, as the following examples illustrate.

\begin{example}
Consider the TSS given by the rules\vspace{-1ex}
\[ a.x \goto{a} x \qquad \frac{x \goto{a} x'}{x\|y \goto{a} x'\|y} \qquad \frac{y \goto{a} y'}{x\|y \goto{a} x\|y'}\]
where $a$ ranges over $A$, and the recursion rule from \df{format} below.
An \emph{infinite trace} of a process $p$ is a sequence $a_1a_2\dots \in A^\omega$
such that there are processes $p_1,p_2,\dots$ with $p \goto{a_1} p_1 \goto{a_2} p_2 \goto{a_3} \dots$.
Let $p \sqsubseteq q$ iff for each infinite trace $\sigma$ of $p$ there is an infinite trace of $q$
that has a suffix in common with $\sigma$. This is a preorder indeed. It is not hard to check that
$\sqsubseteq$ is a precongruence for both action prefixing $a.\_\!\!\_\,$ and parallel composition
$\_\!\!\_ \| \_\!\!\_\,$,
in the sense that (\ref{comp-operators-closed}) holds. However, it fails to be a lean congruence,
because $a.\rec{X| X{=}c.X} \equiv b.\rec{X| X{=}c.X}$, yet when filled in for $Y$ in
$\rec{Z|Z{=}Y\|Z}$ (which can be seen as $!Y$, an infinite parallel composition of copies of $Y$) the two are
no longer equivalent.
\end{example}
I did not find a pair of a TSS and a preorder known from the literature showing the same.
This suggests that most common preorders that are (pre)congruences for a selection of common operators
are also lean (pre)congruences for recursion.

\begin{example}
Consider the TSS with a constant ${\bf 0}$ and action prefixing, and
only the rules for recursion from \df{format} and $a.x \goto{a} x$ for $a \in A$, with $\tau\in A$ the
\emph{internal action}.  Consider any semantic equivalence $\sim$ satisfying $x\sim\tau.x$, and such that
\emph{divergence} $\rec{X|X{=}\tau.X}$ differs from \emph{deadlock} or \emph{inaction} ${\bf 0}$.
Such semantic equivalences are abound in the literature and include the \emph{failures} semantics
of CSP \cite{BHR84,vG93} and \emph{branching bisimilarity with explicit divergence} \cite{GLT09b,vG93}.
They are all lean congruences (at least when no other operators are present).
Yet, since $0 \sim \rec{X|X{=}X} \not\sim \rec{X|X{=}\tau.X}$, they fail to be full congruences.
\end{example}
A lean congruence is required for treating processes as equivalence classes of closed terms
rather than as the closed terms themselves, in such a way that each term
$t\mathbin\in\IT(\Sigma,W)$ with free variables drawn from the set $W$ models a
$W$-ary operator on such processes. As explained in the introduction,
this notion of congruence facilitates a formal comparison of the expressive power of system
description languages~\cite{vG12}. However, it does not allow equivalence preserving modifications
of recursive specifications themselves, as contemplated in the introduction.
That requires a full congruence.

\section{The pure ntyxt/ntyft format with recursion}

\begin{definition}[\emph{ntytt, ntyft, ntyxt, nxytt rules}]\label{def:ntytt}\rm
An \emph{ntytt rule} is a rule in which the right-hand sides of positive premises are variables that are all distinct, and that do not occur in the source. An ntytt rule is an \emph{ntyxt rule} if its source is a variable, an \emph{ntyft rule} if its source contains exactly one function symbol and no multiple occurrences of variables, and an \emph{nxytt rule} if the left-hand sides of its premises are variables. 
\end{definition} 

\noindent
The idea behind the names of the rules is that the `n' in front refers
to the presence of negative premises, and the following four letters
refer to the allowed forms of left- and right-hand sides of
premises and of the conclusion, respectively.
For example, ntyft means a rule with negative premises (n), where
left-hand sides of premises are general terms (t), right-hand sides
of positive premises are variables (y), the source contains exactly
one function symbol (f), and the target is a general term (t).

\begin{definition}\label{df:format}\rm
A TSS is in the \emph{ntyft/ntyxt format with recursion} if for every recursive
specification $S$ and $X \in V_S$ it has a rule\vspace{-1ex}
$$\displaystyle\frac{\rec{S_X|S} \ar{a} z}{\rec{X|S}\ar{a}z}$$
and all of its other rules are ntyft or ntyxt rules.
\end{definition}

\begin{definition}[\emph{Well-founded and pure rules; distance}]\rm\label{def:lookahead}
\hfill The \emph{dependency graph} of an ntytt rule with \plat{$\{t_i\ar{a_i}y_i\mid i\mathbin\in I\}$} as set of
positive premises is the directed graph with edges
$\{\langle x,y_i\rangle\mid x\mathbin\in\var(t_i)\mbox{ for some }i\mathbin\in I\}$.
A ntytt rule is \emph{well-founded} if each backward chain of edges in its dependency graph is finite.
A variable in a rule is \emph{free} if it occurs neither in the source nor in the right-hand sides
of the premises of this rule. A rule is \emph{pure} if it is well-founded and does not contain free variables.
A TSS is \emph{well-founded}, resp.\ \emph{pure}, if all of its rules are.

Let $r\mathbin=\frac{H}{t\goto{a}u}$ be a pure ntytt rule.\vspace{1pt}
The \emph{distance} of a variable $y\in\var(r)$ to the source of $r$ is the ordinal number given by
\begin{tabbing}
       \= $\textit{dist}(x)=0$ \hspace{4.5cm} \= if $x\in\var(t)$,\\
       \> $\textit{dist}(y)=1+\sup(\{\textit{dist}(x) \mid x\mathbin\in\var(t)\})$ \> if $(t\ar{a}y)\in H$.
\end{tabbing}
\end{definition}
{\sc Bol \& Groote} show that bisimilarity is a congruence for any language
specified by a complete TSS in the well-founded ntyft/ntyxt format (without recursion) \cite{BolG96}.
This generalises a result by {\sc Groote} \cite{Gr93}, showing the same for stratified TSSs in the
well-founded ntyft/ntyxt format; here \emph{stratified} is a more restrictive criterion than
completeness, guaranteeing that a TSS has a well-defined meaning as a 2-valued transition relation.
That result, in turn, generalises the congruence formats of {\sc Groote \& Vaandrager} \cite{GrV92}
for the well-founded tyft/tyxt format (obtained by leaving out negative premises) and for the GSOS
format of {\sc Bloom, Istrail \& Meyer} \cite{BIM95}. Both of these generalise the De Simone format \cite{dS84,dS85}.

{\sc Fokkink and van Glabbeek} show that for any complete TSS in tyft/tyxt (resp.\ ntyft/ntyxt) format there exists a
pure (and thus well-founded) complete TSS in tyft (resp.\ ntyft) format that generates the same
transition relation \cite{FG96}.
From this it follows that the restriction to well-founded TSSs can be dropped from the congruence
formats of \cite{BolG96} and \cite{GrV92}.
The result of \cite{FG96} generalises straightforwardly to incomplete TSSs, and to formats with recursion.

\begin{theorem}\label{thm:ntree}
For each TSS in the tyft/tyxt (resp.\ ntyft/ntyxt) format with recursion there exists a pure TSS in
the tyft (resp.\ ntyft) format with recursion, generating the same (3-valued) transition relation.
\end{theorem}

\begin{proof}
  \cite[Theorem~5.4]{FG96} shows that for each TSS $P$ in ntyft/ntyxt format there exists a TSS $P'$ in
  pure ntyft format, such that for any closed transition rule $\frac{N}{\alpha}$ with only negative
  premises, one has $P \vdash \frac{N}{\alpha} \Leftrightarrow P' \vdash \frac{N}{\alpha}$.
  This result generalises seamlessly to TSS in the ntyft/ntyxt format with recursion;
  I leave it to the reader to check that recursion causes no new complications in the proof.

  \cite{FG96} obtains the quoted result for complete TSSs from Thm.~5.4 by means of an
  application of \cite[Prop.~5.2]{FG96}, which says that if $P$ and $P'$ are TSSs such that
  $P \vdash \frac{N}{\alpha} \Leftrightarrow P' \vdash \frac{N}{\alpha}$
  for any closed transition rule $\frac{N}{\alpha}$ with only negative premises,
  then $P$ is complete iff $P'$ is, and in that case they determine the same transition relation.
  This Prop.~5.2 was taken verbatim from \cite[Prop.~29]{vG95}.

  In \cite{vG04}, the journal version of \cite{vG95}, Prop.~29 was extended to also conclude, under
  the same assumption, that $P$ and $P'$ determine the same 3-valued transition relation according
  to the well-founded semantics. Using this version of Prop.~29 instead of Prop.~5.2 yields the
  required result.
\qed
\end{proof}

The next two propositions (not used in the rest of the paper) tell that any language specified by
TSS in the ntyft/ntyxt format with recursion satisfies two sanity requirements from \cite{vG94a}.
The first is that, up to $\equiv_B$, the meaning of a closed term $\rec{X|S}$ is the $X$-component of a solution of $S$:
\begin{proposition}\label{pr:solution bisimilar}
Let $P\mathbin=(\Sigma,R)$ be a TSS in the ntyft/ntyxt format with recursion and $S$
a recursive specification with $X\mathbin\in V_S$. Then $\rec{X|S} \equiv_B \rec{S_X|S}$.
\end{proposition}
\begin{proof}
$P \vdash \rec{X|S} \ar{a} q$ for some $a\in A$ and $q\in\T(\Sigma)$ iff $P \vdash \rec{S_X|S}\ar{a}q$.
\qed
\end{proof}
For the second, \emph{invariance under $\alpha$-conversion},
write $t \eqa u$ if the terms $t,u\in\IT(\Sigma)$ differ only in the
names of their bound variables (the variables from $V_S$ within a
subexpression of the form $\rec{X|S}$).
\begin{proposition}\label{pr:alpha-conversion bisimilar}
Let $P=(\Sigma,R)$ be a TSS in the ntyft/ntyxt format with recursion.
Then $p \eqa q \Rightarrow p \equiv_B q$ for all $p,q\in\T(\Sigma)$.
\end{proposition}
\begin{proof}
By \thm{ntree} I may assume, without loss of generality, that
$P$ is in the pure ntyft format with recursion.
I show that $\eqa$ is a bisimulation on $\T(\Sigma)$---since $\eqa$ is also symmetric, this yields
the required result.\pagebreak

\noindent
Thus I need to show that, for $p,q\in\T(\Sigma)$ and $a\in A$,
\begin{itemize}
\item if $p\eqa q$ and $\PP p \goto{a} p'$, then there is a
$q'$ with $\PP q \goto{a} q'$ and $p' \eqa q'$,
\item if $p\eqa q$ and $\PP q \hoto{a} q'$, then there is a
$p'$ with $\PP p \hoto{a} p'$ and $p' \eqa q'$.
\end{itemize}
To this end it suffices to establish, for all ordinals $\lambda$, that
\begin{enumerate}
\item[4.] if $p\eqa q$ and $P \vdash p \goto{a}_\lambda p'$, then there is a
$q'$ with $\PP q \goto{a} q'$ and $p' \eqa q'$,
\item[2.] if $p\eqa q$ and $\PP q \hoto{a} q'$, then there is a
$p'$ with $P \vdash p \hoto{a}_\lambda p'$ and $p' \eqa q'$.
\end{enumerate}
The  desired result is then obtained by taking $\lambda$ to be the closure
ordinal $\kappa$ of \df{closure ordinal}.
This I will do by induction on $\lambda$, at the same time establishing that
\begin{enumerate}
\item[3.] if $p\eqa q$ and $P \vdash p \gonotto{a}_\lambda$, then $\PP q \gonotto{a}$,
\item[1.] if $p\eqa q$ and $\PP q \honotto{a}$, then $P \vdash p \honotto{a}_\lambda$.
\end{enumerate}
So assume Claims 1--4 have been established for all $\kappa < \lambda$.

Suppose $p\eqa q$ and $\PP q \honotto{a}$.
By Remark~\ref{determined} there is no $q'\mathbin\in\T(\Sigma)$ with \plat{$\PP q \goto{a} q'$}.
So by induction, using Claim 4 above, 
there is no $p'\mathbin\in\T(\Sigma)$ with \plat{$P\vdash p \goto{a}_\kappa p'$}
for some $\kappa < \lambda$. By \df{well-founded} $P \vdash p \honotto{a}_\lambda$. This yields Claim 1.

Now suppose $p\eqa q$ and $\PP q \hoto{a} q'$. I need to find a
$p'$ with $P \vdash q \hoto{a}_\lambda q'$ and $p' \eqa q'$.
This I will do by structural induction on the proof $\pi$ of $p \hoto{a} p'$ from $P$,
making a case distinction based on the shape of $p$.
\begin{itemize}
\item
Let $p=f(p_1,\ldots,p_n)$. Then $q=f(q_1,\ldots,q_n)$ where $p_i \eqa q_i$ for $i\mathbin=1,\ldots,n$.
Let $\pi$ be a proof of \plat{$\PP q\hoto{a} q'$} from $P$.
By Defs.~\ref{df:proof} and~\ref{df:format}, there must be a pure ntyft rule
$r=\frac{H}{f(x_1,...,x_n)\goto{a}t}$ in $R$ and a closed substitution
$\nu$ with $\nu(x_i)\mathbin=q_i$ for $i\mathbin=1,...,n$ and $t[\nu]\mathbin=q'$,
such that for each \plat{$(t_y\ar{c}y)\in H$} the transition \plat{$\PP t_y[\nu]\hoto{c}\nu(y)$}
is provable from $P$ by means of a strict subproof of $\pi$,
and $\PP u[\nu]\honotto{c}$ for each \plat{$(u\nar{c})\in H$}.
Next, I define a substitution $\sigma:\var(r)\rightarrow\T(\Sigma)$ such that
\begin{enumerate}[(i), leftmargin=*]
\item $\sigma(x_i)=p_i$ for $i=1,\ldots,n$,
\item $\sigma(y)\eqa \nu(y)$ for each $y\mathbin\in \var(r)$,
\item \plat{$P\vdash t_y[\sigma]\hoto{c}_\lambda \sigma(y)$} for each \plat{$(t_y\ar{c}y)\in H$}.
\end{enumerate}
The definition of $\sigma(y)$ and the inference of (i)--(iii) above proceed with induction on the
distance of $y\mathbin\in\var(r)$ from the source of $r$,
\vspace{1ex}

{\it Base case:\/} Let $\sigma(x_i):=p_i$ for $i=1,\ldots,n$, so that Property (i) is satisfied.
Regarding Property (ii), $\sigma(x_i)\eqa \nu(x_i)$ for $i=1,\ldots,n$.
\vspace{1ex}

{\it Induction step:\/} When defining $\sigma(y)$ for some $y\mathbin\in\Var$ with \plat{$(t_y\ar{c}y)\in H$},
by induction $\sigma(x)$ has been defined already for all $x\mathbin\in\var(t_y)$, so I may assume that
$\sigma(x)\eqa \nu(x)$ for all $x\mathbin\in\var(t_y)$ and hence $t_y[\sigma] \eqa t_y[\nu]$.

By induction on $\pi$, there is a $p_y$ with
\plat{$P\vdash t_y[\sigma]\hoto{c}_\lambda p_y$} and $p_y \eqa \nu(y)$.
Define $\sigma(y):=p_y$. Properties (ii) and (iii) now hold for $y$.
\vspace{1ex}

Take $p':=t[\sigma]$. So $p'\mathbin=t[\sigma]\eqa t[\nu] \mathbin= q'$ by Property (ii) of $\sigma$.
For each premise \plat{$(u\nar{c})\in H$} one has $u[\sigma]\eqa u[\nu]$ by Property (ii) of $\sigma$.
So \plat{$P\vdash u[\sigma]\honotto{c}_\lambda$} by Claim 1.
By Defs.~\ref{df:proof} and~\ref{df:format}, together with Property (iii) of $\sigma$, this implies
\plat{$P\vdash p=f(p_1,\ldots,p_n) \hoto{a}_\lambda t[\sigma] =p'$}.

\item
Let $p=\rec{X|S}$.  Then $q=\rec{\alpha(X) | S'[\alpha]}$ for some recursive specification
$S':V_S\rightarrow\IT(\Sigma)$ with $S_Y \eqa S'_Y$ for all $Y\in V_S$,
and an injective substitution $\alpha:V_S \rightarrow \Var$ such that the range of $\alpha$ contains
no variables occurring free in $\rec{S'_Y|S}$ for some $Y\in V_S$.
Now $\rec{S_X|S} \eqa \rec{S_X|S'} \eqa \rec{S'_{\alpha(X)} | S'[\alpha]}$.
Let $\pi$ be a proof of \plat{$\PP q\hoto{a} q'$} from $P$.
By Defs.~\ref{df:proof} and~\ref{df:format}
\plat{$\PP \rec{S'_{\alpha(X)}|S'[\alpha]} \hoto{a} q'$} is provable from $P$ by means of a strict subproof of $\pi$.
So by induction there is a $p'$ such that \plat{$P\vdash \rec{S_X|S} \hoto{a}_\lambda p'$} and $p' \eqa q'$.
By Defs.~\ref{df:proof} and~\ref{df:format}, $P\vdash p = \rec{X|S} \hoto{a}_\lambda p'$.
\end{itemize}
This establishes Claim 2.

Next, suppose that $p\eqa q$ and $P \vdash p \gonotto{a}_\lambda$.
By \df{well-founded} there is no $p'\mathbin\in\T(\Sigma)$ with \plat{$P \vdash p \hoto{a}_\lambda p'$}.
Using Claim 2, there is no $q'\mathbin\in\T(\Sigma)$ with \plat{$\PP q \hoto{a} q'$}.
By Remark~\ref{determined}, \plat{$\PP q \gonotto{a}$}. This yields Claim 3.

Claim 4 follows by structural induction on the proof of $p \ar{a}_\lambda p'$ from $P$, pretty much
in the same way as Claim 2 above.
\qed
\end{proof}
\pr{alpha-conversion bisimilar} could be classified as ``self-evident''.
One reason to spell out the proof above is to obtain a template for bisimilarity proofs in
the setting of the well-founded semantics.
I will use this template in the forthcoming lean congruence proof.

\section{A lean congruence result}

The following congruence proof is strongly inspired by the one in \cite{BolG96}.

\begin{theorem}\label{thm:congruence}
Bisimilarity is a lean precongruence for any language specified by a TSS in the ntyft/ntyxt format
with recursion.
\end{theorem}

\begin{trivlist} \item[\hspace{\labelsep}\bf Proof:]
By \thm{ntree} I may assume, without loss of generality, that
$P=(\Sigma,R)$ is a TSS in the pure ntyft format with recursion.
Let $\R$ be the smallest lean precongruence containing bisimilarity, i.e.,
$\mathord{\R}\subseteq \T(\Sigma) \times \T(\Sigma)$
is the smallest relation on processes satisfying
\begin{itemize}
\item if $p \sqsubseteq_B q$ then $p \R q$,
\item if $(f,n)\mathbin\in\Sigma$ and $p_i \mathbin{\R} q_i$ for all $i\mathbin=1,...,n$,
      then $f(p_1,\ldots,p_n)\R f(q_1,\ldots,q_n)$,
\item and if $S:V_S \rightarrow \IT(\Sigma)$ with $Z \in V_S \subseteq \Var$,
      and $\rho,\nu:\Var\setminus V_S \rightarrow\T(\Sigma)$ satisfy $\rho(x) \R \nu(x)$ for all
      $x\mathbin\in\Var\setminus V_S$, then $\rec{Z|S}[\rho] \R \rec{Z|S}[\nu]$.
\end{itemize}
A trivial structural induction on $t\mathbin\in\IT(\Sigma)$, using the last two clauses, shows
that if $\rho,\nu\!:\!\Var \mathbin\rightarrow\T(\Sigma)$ satisfy $\rho(x) \mathbin{\R} \nu(x)$
for all $x\mathbin\in\var(t)$, then $t[\rho] \R t[\nu]$.\hfill ({\color{red}*})\\
As $\rec{\_|S}[\rho]:V_S\rightarrow\T(\Sigma)$ and $\rec{\_|S}[\nu]:V_S\rightarrow\T(\Sigma)$,
this implies that in the last clause one even has
$\rec{t|S}[\rho] \R \rec{t|S}[\nu]$ for all terms $t\in\IT(\Sigma,V_S)$.\hfill (\$)

It suffices to show that $\R$ is a bisimulation, because this implies
${\R}\subseteq{\sqsubseteq_B}$, so that $\R$ equals $\sqsubseteq_B$, and
({\color{red}*}) says that $\R$ is a lean precongruence.
Thus I need to show that, for $p,q\in\T(\Sigma)$ and $a\in A$,
\begin{itemize}
\item if $p\R q$ and $\PP p \goto{a} p'$, then there is a
$q'$ with $\PP q \goto{a} q'$ and $p' \BR q'$,
\item if $p\R q$ and $\PP q \hoto{a} q'$, then there is a
$p'$ with $\PP p \hoto{a} p'$ and $p' \mathrel{\color{red}\R} q'$.
\end{itemize}
To this end it suffices to establish, for all ordinals $\lambda$, that
\begin{enumerate}
{\color{darkgreen} \item[4.] if $p\R q$ and $P \vdash p \goto{a}_\lambda p'$, then there is a
$q'$ with $\PP q \goto{a} q'$ and $p' \BR q'$,}
{\color{darkred} \item[2.] if $p\R q$ and $\PP q \hoto{a} q'$, then there is a
$p'$ with $P \vdash p \hoto{a}_\lambda p'$ and $p' \mathrel{\color{red}\R} q'$}.
\end{enumerate}
The  desired result is then obtained by taking $\lambda$ to be the closure
ordinal $\kappa$ of \df{closure ordinal}.
This I will do by induction on $\lambda$, at the same time establishing that
\begin{enumerate}
\item[3.] if $p\R q$ and $P \vdash p \gonotto{a}_\lambda$, then $\PP q \gonotto{a}$,
\item[1.] if $p\R q$ and $\PP q \honotto{a}$, then $P \vdash p \honotto{a}_\lambda$.
\end{enumerate}
So assume Claims 1--4 have been established for all $\kappa < \lambda$.

Suppose $p\R q$ and $\PP q \honotto{a}$.
By Remark~\ref{determined} there is no $q'\mathbin\in\T(\Sigma)$ with \plat{$\PP q \goto{a} q'$}.
So by induction, using {\color{darkgreen}Claim 4} above, 
there is no $p'\mathbin\in\T(\Sigma)$ with \plat{$P\vdash p \goto{a}_\kappa p'$}
for some $\kappa < \lambda$. By \df{well-founded} $P \vdash p \honotto{a}_\lambda$. This yields Claim 1.

{\color{darkred}%
Now suppose $p\R q$ and $\PP q \hoto{a} q'$. I need to find a
$p'$ with $P \vdash p \hoto{a}_\lambda p'$ and $p' \BRR q'$.
This I will do by structural induction on the proof $\pi$ of $q \hoto{a} q'$ from $P$.
I make a case distinction based on the derivation of $p\R q$.
\begin{itemize}
\item
Let $p\sqsubseteq_B q$.
Using that $\sqsubseteq_B$ is a bisimulation, there must be a process $p'$ such
that $\PP p \hoto{a} p'$ and $p'\sqsubseteq_B q'$, hence $p' \BRR q'$.
Since $\PP p \hoto{a} p'$, certainly \plat{$P\vdash p\hoto{a}_\lambda p'$}, by Remark~\ref{union}.
\item
Let $p=f(p_1,\ldots,p_n)$ and
$q=f(q_1,\ldots,q_n)$ where $p_i \R q_i$ for $i=1,\ldots,n$.
Let $\pi$ be a proof of \plat{$q\hoto{a} q'$} from $P$.
By Defs.~\ref{df:proof} and~\ref{df:format}, there must be a pure ntyft rule
$r=\frac{H}{f(x_1,...,x_n)\goto{a}t}$ in $R$ and a closed substitution
$\nu$ with $\nu(x_i)\mathbin=q_i$ for $i\mathbin=1,...,n$ and $t[\nu]\mathbin=q'$,
such that for each \plat{$(t_y\ar{c}y)\in H$} the transition \plat{$t_y[\nu]\hoto{c}\nu(y)$}
is provable from $P$ by means of a strict subproof of $\pi$,
and $\PP u[\nu]\honotto{c}$ for each \plat{$(u\nar{c})\in H$}.
Next, I define a substitution $\sigma:\var(r)\rightarrow\T(\Sigma)$ such that
\begin{enumerate}[(i), leftmargin=*]
\item $\sigma(x_i)=p_i$ for $i=1,\ldots,n$,
\item $\sigma(y)\BRR \nu(y)$ for each $y\mathbin\in \var(r)$,
\item \plat{$P\vdash t_y[\sigma]\hoto{c}_\lambda \sigma(y)$} for each \plat{$(t_y\ar{c}y)\in H$}.
\end{enumerate}
The definition of $\sigma(y)$ and the inference of (i)--(iii) above proceed with induction on the
distance of $y\mathbin\in\var(r)$ from the source of $r$,
\vspace{1ex}

{\it Base case:\/} Let $\sigma(x_i):=p_i$ for $i=1,\ldots,n$, so that Property (i) is satisfied.
Regarding Property (ii), $\sigma(x_i)\R \nu(x_i)$ for $i=1,\ldots,n$.
\vspace{1ex}

{\it Induction step:\/} When defining $\sigma(y)$ for some $y\mathbin\in\Var$ with \plat{$(t_y\ar{c}y)\in H$},
by induction $\sigma(x)$ has been defined already for all $x\mathbin\in\var(t_y)$, so {I may assume that
$\sigma(x)\BRR \nu(x)$ for all $x\mathbin\in\var(t_y)$ and hence $t_y[\sigma] \BRR t_y[\nu]$ by ({\color{red}*}).

By induction on $\pi$, there is a $p_y$ with
\plat{$P\vdash t_y[\sigma]\hoto{c}_\lambda p_y$} and $p_y \BRR \nu(y)$.}
Define $\sigma(y):=p_y$. Properties (ii) and (iii) now hold for $y$.
\vspace{1ex}

Take $p':=t[\sigma]$. So $p'\mathbin=t[\sigma]\BRR t[\nu] \mathbin= q'$ by ({\color{red}*}) and Property (ii) of $\sigma$.
For each premise \plat{$(u\nar{c})\in H$} one has $u[\sigma]\BRR u[\nu]$ by ({\color{red}*}) and Property (ii) of $\sigma$.
So \plat{$P\vdash u[\sigma]\honotto{c}_\lambda$} by Claim 1.
By Defs.~\ref{df:proof} and~\ref{df:format}, together with Property (iii) of $\sigma$, this implies\\
\plat{$P\vdash p=f(p_1,\ldots,p_n) \hoto{a}_\lambda t[\sigma] =p'$}.
\item
Let $p\mathbin=\rec{Z|S}[\rho]\mathbin=\rec{Z|S[\rho]}$ and $q\mathbin=\rec{Z|S}[\sigma]\mathbin=\rec{Z|S[\sigma]}$
where $S:V_S \rightarrow \IT(\Sigma)$
with $Z \in V_S \subseteq \Var$, $\rho,\sigma\!:\Var{\setminus} V_S \mathbin\rightarrow\T(\Sigma)$,
and for all $x\in\Var\setminus V_S$ one has $\rho(x) \R \sigma(x)$.
Let $\pi$ be a proof of \plat{$q\hoto{a} q'$} from $P$.
By Defs.~\ref{df:proof} and~\ref{df:format}
\plat{$\rec{S_Z|S[\sigma]} \hoto{a} q'$} is provable from $P$ by means of a strict subproof of $\pi$.
By (\$) above one has $\rec{S_Z|S[\rho]} \R \rec{S_Z|S[\sigma]}$.
So by induction there is a $p'$ such that \plat{$P\vdash \rec{S_Z|S[\rho]} \hoto{a}_\lambda p'$} and $p' \BRR q'$.
By Defs.~\ref{df:proof} and~\ref{df:format}, $P\vdash p = \rec{Z|S[\sigma]} \hoto{a}_\lambda p'$.
\end{itemize}}

Next, suppose that $p\R q$ and $P \vdash p \gonotto{a}_\lambda$.
By \df{well-founded} there is no $p'\mathbin\in\T(\Sigma)$ with \plat{$P \vdash p \hoto{a}_\lambda p'$}.
Using {\color{darkred}Claim 2}, there is no $q'\mathbin\in\T(\Sigma)$ with \plat{$\PP q \hoto{a} q'$}.
By Remark~\ref{determined}, \plat{$\PP q \gonotto{a}$}. This yields Claim 3.

\hypertarget{Claim4}{\color{darkgreen}%
Finally, suppose $p\R q$ and $P \vdash p \goto{a}_\lambda p'$. I need to find a
$q'$ with $\PP q \goto{a} q'$ and $p' \BR q'$.} \color{darkgreen}%
This I will do by structural induction on the proof $\pi$ of $p \ar{a}_\lambda p'$ from $P$.
I make a case distinction based on the derivation of $p\R q$.
\begin{itemize}
\item
Let $p\sqsubseteq_B q$. Since $P \vdash p \goto{a}_\lambda p'$, certainly $\PP p\goto{a} p'$, by Remark~\ref{union}.
Using that $\sqsubseteq_B$ is a bisimulation, there must be a process $q'$ such
that $\PP q \ar{a} q'$ and $p'\sqsubseteq_B q'$, hence $p' \BR q'$.
\item
Let $p\mathbin=f(p_1,\ldots,p_n)$ and
$q\mathbin=f(q_1,\ldots,q_n)$ where $p_i \R q_i$ for $i=1,\ldots,n$.
Let $\pi$ be a proof of \plat{$p\ar{a}_\lambda p'$} from $P$.\linebreak[4]
By Defs.~\ref{df:proof},~\ref{df:well-founded} and~\ref{df:format}, there must be a pure ntyft rule
$r=\frac{H}{f(x_1,...,x_n)\goto{a}t}$ in $R$ and a closed substitution
$\sigma$ with $\sigma(x_i)\mathbin=p_i$ for $i\mathbin=1,...,n$ and $t[\sigma]\mathbin=p'$,
such that for each \plat{$(t_y\ar{c}y)\in H$} the transition \plat{$t_y[\sigma]\mathbin{\ar{c}_\lambda}\sigma(y)$}
is provable from $P$ by means of a strict subproof of $\pi$,
and $P\vdash u[\sigma]\nar{c}_\lambda$ for each \plat{$(u\nar{c})\in H$}.
Next, I define a substitution $\nu\!:\!\var(r)\mathbin\rightarrow\T(\Sigma)$ such that
\begin{enumerate}[(i), leftmargin=*]
\item $\nu(x_i)=q_i$ for $i=1,\ldots,n$,
\item $\sigma(y)\BR \nu(y)$ for each $y\mathbin\in \var(r)$,
\item \plat{$\PP t_y[\nu]\ar{c} \nu(y)$} for each \plat{$(t_y\ar{c}y)\in H$}.
\end{enumerate}
The definition of $\nu(y)$ and the inference of (i)--(iii) above proceed with induction on the distance of
$y\mathbin\in\var(r)$ from the source of $r$,
\vspace{1ex}

{\it Base case:\/} Let $\nu(x_i):=q_i$ for $i=1,\ldots,n$, so that Property (i) is satisfied.
Regarding Property (ii), $\sigma(x_i)\R \nu(x_i)$ for $i=1,\ldots,n$.
\vspace{1ex}

{\it Induction step:\/} When defining $\nu(y)$ for some $y\mathbin\in\Var$ with \plat{$(t_y\ar{c}y)\in H$},
by induction $\nu(x)$ has been defined already for all $x\mathbin\in\var(t_y)$, so I may assume that
$\sigma(x)\BR \nu(x)$ for all $x\mathbin\in\var(t_y)$  {\color{blue}and hence $t_y[\sigma] \BR t_y[\nu]$ by ({\color{red}*})}.

By induction on $\pi$, {\color{blue}there is a $q_y$ with
\plat{$\PP t_y[\nu]\ar{c} q_y$} and $\sigma(y) \BR q_y$}.
Define $\nu(y):=q_y$. Properties (ii) and (iii) now hold for $y$.
\vspace{1ex}

Take $q':=t[\nu]$. So $p'\mathbin=t[\sigma]\BR t[\nu] \mathbin= q'$ by ({\color{red}*}) and Property (ii) of $\nu${\color{blue}.}
For each premise \plat{$(u\nar{c})\in H$} one has $u[\sigma]\BR u[\nu]$ by ({\color{red}*}) and Property (ii) of $\nu$.
So \plat{$\PP u[\nu]\nar{c}$} by Claim 3.
Since $CT^+$ is closed under deduction, together with Property (iii) of $\nu$ this implies 
\plat{$\PP q=f(q_1,\ldots,q_n) \ar{a} t[\nu] =q'$}.
\item
Let $p\mathbin=\rec{Z|S}[\rho]\mathbin=\rec{Z|S[\rho]}$ and $q\mathbin=\rec{Z|S}[\nu]\mathbin=\rec{Z|S[\nu]}$
where $S:V_S \rightarrow \IT(\Sigma)$
with $Z \in V_S \subseteq \Var$, $\rho,\nu\!:\Var{\setminus} V_S \mathbin\rightarrow\T(\Sigma)$,
and for all $x\in\Var\setminus V_S$ one has $\rho(x) \R \nu(x)$.
Let $\pi$ be a proof of \plat{$p\ar{a}_\lambda p'$} from $P$.
By Defs.~\ref{df:proof},~\ref{df:well-founded} and~\ref{df:format}
\plat{$\rec{S_Z|S[\rho]} \ar{a}_\lambda p'$} is provable from $P$ by means of a strict subproof of $\pi$.
By (\$) above one has $\rec{S_Z|S[\rho]} \R \rec{S_Z|S[\nu]}$.
So by induction there is a $q'$ such that \plat{$\PP \rec{S_Z|S[\nu]} \ar{a} q'$} and $p' \BR q'$.
By Defs.~\ref{df:proof} and~\ref{df:format}, $\PP q = \rec{Z|S[\nu]} \ar{a} q'$.
\end{itemize}
This yields Claim 4.\hfill $\Box$
\end{trivlist}
The above result implies that any ntyft/ntyxt language with recursion satisfies
congruence requirement (\ref{comp-operators-closed}) up to $\sqsubseteq_B$, but is not strong
enough to yield (\ref{comp-recursion-closed}).

\section{A full congruence result}

In this section I deal with positive TSSs only.
Here the relations $\hoto{a}_\lambda$ and $\goto{a}_\mu$ for ordinals $\lambda$ and $\mu$ all
coincide, and ${\sqsubseteq_B} = {\equiv_B}$.
The following auxiliary concept was used in \cite{Mi90ccs} to show that CCS satisfies
Condition (\ref{comp-recursion-closed}) of \df{full congruence}.

\begin{definition}\label{df:upto}\rm
A symmetric relation $\mathord{\R} \subseteq \T(\Sigma)\times\T(\Sigma)$ is a \emph{bisimulation up to $\sim$} if
$p\R q$ and \plat{$\PP p \goto{a} p'$} imply that there is a $q'$ with \plat{$\PP q \goto{a} q'$} and $p' \sim \R \sim q'$,
for all $a\in A$.
Here ${\sim \R \sim} := \{(r,s)\!\mid \exists r'\!,s'\!.~ r \sim r' \R s' \sim s\}$.
\end{definition}

\begin{proposition}[\cite{Mi90ccs}]\label{pr:upto}
If $p \R q$ for some bisimulation $\R$ up to $\equiv_B$, then $p \equiv_B q$.
\end{proposition}
\begin{proof}
Using the reflexivity of $\equiv_B$ it suffices to show that $\equiv_B \R \equiv_B$ is a bisimulation.
Using symmetry and transitivity of $\equiv_B$ this is straightforward.
\qed
\end{proof}

\begin{theorem}\label{thm:compositionality positive}
Bisimilarity is a full congruence for any language specified by a TSS in the tyft/tyxt format
with recursion.
\end{theorem}

\begin{trivlist} \item[\hspace{\labelsep}\bf Proof:]
By \thm{ntree} I may assume, without loss of generality, that
$P=(\Sigma,R)$ is a TSS in the pure tyft format with recursion.
Let $S,S':W \rightarrow \IT(\Sigma,W)$ be recursive specifications with
$S_Y[\sigma] \equiv_B S'_Y[\sigma]$ for all $Y\in W$ and $\sigma:W\rightarrow \T(\Sigma)$.$^{\mbox{\scriptsize\ref{onlytwo}}}$
I need to show that $\rec{X|S}\equiv_B\rec{X|S'}$ for all $X\in W$.
Let $\mathord{\R}\subseteq \T(\Sigma) \times \T(\Sigma)$
be the smallest relation on processes satisfying
\begin{itemize}
\item $\rec{X|S} \R \rec{X|S'}$ and $\rec{X|S'} \R \rec{X|S}$ for all $X\in W$,
\item if $(f,n)\mathbin\in\Sigma$ and $p_i \mathbin{\R} q_i$ for all $i\mathbin=1,...,n$,
      then $f(p_1,\ldots,p_n)\R f(q_1,\ldots,q_n)$,
\item and if ${S''}:V_{S''} \rightarrow \IT(\Sigma)$ with $Z \in V_{S''} \subseteq \Var$,
      and $\rho,\nu:\Var\setminus V_{S''} \rightarrow\T(\Sigma)$ satisfy $\rho(x) \R \nu(x)$ for all
      $x\in\Var\setminus V_{S''}$, then $\rec{Z|{S''}}[\rho] \R \rec{Z|{S''}}[\nu]$.
\end{itemize}
A trivial structural induction on $t\mathbin\in\IT(\Sigma)$, using the last two clauses, shows
that if $\rho,\nu:\Var \rightarrow\T(\Sigma)$ satisfy $\rho(x) \R \nu(x)$ for all $x\in\Var$,
then $t[\rho] \R t[\nu]$.\hfill ({\color{red}*})\\
So in the first clause one even has
$\rec{t|S} \R \rec{t|S'}$ for all $t\in\IT(\Sigma,W)$,\hfill (\#)\\
and in the last clause $\rec{t|{S''}}[\rho] \R \rec{t|{S''}}[\nu]$ for all $t\in\IT(\Sigma,V_{S''})$.\hfill (\$)

It suffices to show that $\R$ is a bisimulation up to $\equiv_B$, because with \pr{upto} this implies ${\R}\subseteq{\equiv_B}$.
By construction $\R$ is symmetric. So it suffices to show that,
\begin{center}
if $p\R q$ and $P \vdash p \ar{a} p'$, then there is a
$q'$ with $P \vdash q \ar{a} q'$ and $p' \R \equiv_B q'$,
\end{center}
for all $p,q\in\T(\Sigma)$ and $a\in A$,
This I will do by structural induction on the proof $\pi$ of $p \ar{a} p'$ from $P$.
I make a case distinction based on the derivation of $p\R q$.

\begin{itemize}
\item
Let $p=\rec{X|S}$ and $q=\rec{X|S'}$ with $X \mathbin\in W$.
Let $\pi$ be a proof of \plat{$p\ar{a}p'$} from $P$.
By Definitions~\ref{df:proof} and~\ref{df:format}
\plat{$\rec{S_X|S} \ar{a} p'$} is provable from $P$ by means of a strict subproof of $\pi$.
By (\#) above one has $\rec{S_X|S} \R \rec{S_X|S'}$.
So by induction there is an $r'$ such that \plat{$P\vdash \rec{S_X|S'} \ar{a} r'$} and $p' \R\equiv_B r'$.
Since $\rec{\_\!\_ | S'}$ is a substitution of the form  $\sigma:W \rightarrow\T(\Sigma)$,
one has $\rec{S_X|S'} \equiv_B \rec{S'_X|S'}$.
Hence there is a $q'$ such that \plat{$P\vdash \rec{S'_X|S'} \ar{a} q'$} and $r' \equiv_B q'\!$.
So $p' \R \equiv_B q'\!$.
By Definitions~\ref{df:proof} and~\ref{df:format} $P\vdash q \mathbin= \rec{X|S'} \mathbin{\ar{a}} q'\!$.
\item
The case $p=\rec{X|S'}$ and $q=\rec{X|S}$ goes likewise, swapping the r\^oles of $S'_X$ and $S_X$,
and using the substitution $\rec{\_\!\_ | S}$.\,\footnote{\label{onlytwo}%
This proof shows that in the full congruence property (\ref{comp-recursion-closed}) one only
needs to assume $S_Y[\sigma] \equiv_B S'_Y[\sigma]$ for two specific substitutions $\sigma$:
namely $\sigma(Y):=\rec{Y|S'}$, resp.~$\rec{Y|S}$.}
\item
The remaining two cases proceed in the same way as in \hyperlink{Claim4}{\color{darkgreen}the proof of Claim~4} for
\thm{congruence}, but suppressing $\lambda$ and with $\R\equiv_B$ substituted for the blue occurrences of $\R$. In the last
case there are no further changes, so I will not repeat it here. The remaining case needs a few
elaborations---these involve the blue coloured segments in the proof of Claim~4:
\item
Let $p\mathbin=f(p_1,\ldots,p_n)$ and
$q\mathbin=f(q_1,\ldots,q_n)$ where \mbox{$p_i \R q_i$} for $i=1,\ldots,n$.
Let $\pi$ be a proof of \plat{$p\ar{a} p'$} from $P$.\linebreak[4]
By Defs.~\ref{df:proof} and~\ref{df:format}, there must be a pure tyft rule
$r=\frac{H}{f(x_1,...,x_n)\goto{a}t}$ in $R$ and a closed substitution
$\sigma$ with $\sigma(x_i)\mathbin=p_i$ for $i\mathbin=1,...,n$ and $t[\sigma]\mathbin=p'$,
such that for each \plat{$(t_y\ar{c}y)\in H$} the transition \plat{$t_y[\sigma]\mathbin{\ar{c}}\sigma(y)$}
is provable from $P$ by means of a strict subproof of $\pi$.
Next, I define a substitution $\nu\!:\!\var(r)\mathbin\rightarrow\T(\Sigma)$ such that
\begin{enumerate}[(i), leftmargin=*]
\item $\nu(x_i)=q_i$ for $i=1,\ldots,n$,
\item $\sigma(y)\R\equiv_B \nu(y)$ for each $y\mathbin\in \var(r)$,
\item \plat{$P\vdash t_y[\nu]\ar{c} \nu(y)$} for each \plat{$(t_y\ar{c}y)\in H$}.
\end{enumerate}
The definition of $\nu(y)$ and the inference of (i)--(iii) above proceed with induction on the distance of
$y\mathbin\in\var(r)$ from the source of $r$,
\vspace{1ex}

{\it Base case:\/} Let $\nu(x_i):=q_i$ for $i=1,\ldots,n$, so that Property (i) is satisfied.
Regarding Property (ii), $\sigma(x_i)\R \nu(x_i)$ for $i=1,\ldots,n$.
\vspace{1ex}

{\it Induction step:\/} When defining $\nu(y)$ for some $y\mathbin\in\Var$ with \plat{$(t_y\ar{c}y)\in H$},
by induction $\nu(x)$ has been defined already for all $x\mathbin\in\var(t_y)$, so I may assume that
  $\sigma(x)\R\equiv_B \nu(x)$ for all $x\mathbin\in\var(t_y)$, i.e., there exists a substitution
  $\rho\!:\!\var(r)\mathbin\rightarrow\T(\Sigma)$ with
  $\sigma(x)\mathbin{\R}\rho(x)\mathbin{\equiv_B} \nu(x)$ for all $x\mathbin\in\var(t_y)$.
  Now $t_y[\sigma] \mathbin{\R} t_y[\rho]$ by ({\color{red}*})
  and $t_y[\rho] \mathbin{\equiv_B} t_y[\nu]$ by \thm{congruence}.

  By induction on $\pi$, there is an $r_y$ with
  \plat{$P\vdash t_y[\rho]\ar{c} r_y$} and $\sigma(y) \R\equiv_B r_y$.
  By the definition of bisimilarity, there is a $q_y$ with
  \plat{$P\vdash t_y[\nu]\ar{c} q_y$} and $r_y \equiv_B q_y$.
Define $\nu(y):=q_y$. Properties (ii) and (iii) now hold for $y$.
\vspace{1ex}

Take $q':=t[\nu]$. So $p'\mathbin=t[\sigma]\R\equiv t[\nu] \mathbin= q'$ by ({\color{red}*}),
Property (ii) of $\nu$, and \thm{congruence}.
By Defs.~\ref{df:proof} and~\ref{df:format}, together with Property (iii) of $\nu$, this implies\\
\plat{$P \vdash q=f(q_1,\ldots,q_n) \ar{a} t[\nu] =q'$}.
\hfill $\Box$
\end{itemize}
\end{trivlist}
It remains an open question whether the above result can be generalised to the ntyft/ntyxt format
with recursion. A direct combination of the proofs of Thms.~\ref{thm:congruence}
and~\ref{thm:compositionality positive} does not work, however. An attempt in this direction would
substitute either $\R\sqsubseteq_B$ or $\sqsubseteq_B\R$ for the red $\color{red}\R$ in
{\color{darkred}Claim 2} in the proof of \thm{congruence}.
Both attempts fail on the case $p=\rec{X|S}$ and $q=\rec{X|S'}$ in the proof of \thm{compositionality positive}.

The first attempt would from $P\vdash\rec{S'_X|S'}\hoto{a}q'$ infer $P\vdash\rec{S_X|S'}\hoto{a}r'$
by bisimilarity, and then infer $P\vdash\rec{S_X|S}\hoto{a}_\lambda p'$ by induction. However, one may not
use induction, as the transition $\rec{S_X|S'}\hoto{a}r'$ may be derived later than $\rec{X|S'}\hoto{a}q'$.
In fact, if a variant of this approach would work, skipping $\rec{X|S'} \R \rec{X|S}$ from the
definition of $\R$, one could prove a false version of (\ref{comp-recursion-closed}) that assumes the
antecedent only for the single substitution $\rec{\_\!\_ | S'}$ (cf.~Footnote~\ref{onlytwo});
it is trivial to find a counterexample in the GSOS format with unguarded recursion.

The second attempt would from $P\vdash\rec{S'_X|S'}\hoto{a}q'$ infer $P\vdash\rec{S'_X|S}\hoto{a}_\lambda r'$
by induction, and then $P\vdash\rec{S_X|S}\hoto{a}_\lambda p'$ by bisimilarity. The latter step is invalid,
as $\rec{S_X|S'}\hoto{a}_\lambda r'$ is only an overapproximation of $P\vdash \rec{S_X|S'}\hoto{a} r'$.

\bibliographystyle{eptcs}

\end{document}